\documentclass[letterpaper, 10 pt, conference]{ieeeconf}  

\IEEEoverridecommandlockouts                              

\usepackage{soul}
\usepackage{amsmath,amssymb,amsfonts}
\usepackage{graphicx}
\usepackage{booktabs}
\usepackage{algorithm}
\usepackage{algpseudocode}
\usepackage{xcolor}
\usepackage{cite}
\usepackage{hyperref}
\newcommand{\Xone}{\mathrm{X}_1}
\newcommand{\Proj}[1]{\Pi_{\Xone}\!\left(#1\right)}

\newtheorem{theorem}{Theorem}  

\newtheorem{innerassump}{Assumption}
\newenvironment{assump}[1]
  {\renewcommand\theinnerassump{#1.\arabic{innerassump}}\innerassump}
  {\endinnerassump}

\title{\LARGE \bf
Asymmetric Nash Seeking via Best--Response Maps: Global Linear Convergence and Robustness to Inexact Reaction Models
}

\author{Mahdis Rabbani$^{1}$, Navid Mojahed$^{1}$, and Shima Nazari$^{1}$
\thanks{$^{1}$Authors are with Department of Mechanical and Aerospace Engineering, University of California, Davis, 1 Shields Ave, Davis, CA 95616.
        {\tt\small \{mrabbani, nmojahed, snazari\}@ucdavis.edu}}%
}

\begin{document}

\maketitle
\thispagestyle{empty}
\pagestyle{empty}

\newcommand{\R}{\mathbb{R}}
\newcommand{\norm}[1]{\left\lVert #1 \right\rVert}
\newcommand{\grad}{\nabla}
\newcommand{\argmin}{\operatorname*{arg\,min}}
\newcommand{\BR}{\mathsf{BR}}
\newcommand{\prox}{\operatorname*{prox}}

\begin{abstract}

Nash equilibrium computation in multi-agent decision-making can be challenging when a decision-maker has explicit access to its own optimization model but not to that of another agent. This paper studies a class of two-player constrained games with decoupled feasible sets under an asymmetric information structure: Player~1 knows its own objective and constraints but represents Player~2's response behavior through a best-response correspondence, a single-valued response map, or an approximation thereof. The equilibrium notion remains the standard Nash equilibrium, and the asymmetry concerns only the information available for computing it. We propose a projected gradient-response iteration that does not require explicit knowledge of Player~2's objective or optimality conditions. Under suitable regularity conditions, we establish the existence of a Nash equilibrium under set-valued best-response information, uniqueness under a single-valued Lipschitz continuous response map, and global linear convergence of the proposed iteration. We further analyze the inexact case and show that a uniformly bounded response-map error of size $\varepsilon$ yields an explicit $O(\varepsilon)$ ultimate bound around the Nash equilibrium associated with the exact response map. Numerical results corroborate the predicted convergence behavior and error scaling.

\end{abstract}

\section{Introduction}
\label{sec:introduction}

Strategic decision-making in multi-agent systems is central to autonomy, robotics, and networked control, where agents pursue individual objectives while influencing one another through coupled decisions and shared environments
\cite{ChenRen2019MASSurvey,ZhangEtAl2020NCSurvey,WangEtAl2019GameTheoreticPlanning,TurnwaldEtAl2016HumanAvoidance}.
In such settings, Nash equilibrium (NE) provides a principled notion of interaction-consistent behavior and underlies a broad class of game-theoretic planning and control methods
\cite{Rosen1965,WangEtAl2019GameTheoreticPlanning,TurnwaldEtAl2016HumanAvoidance}.
This viewpoint is especially appealing in interaction-aware autonomy, where mutual adaptation arises naturally in scenarios such as lane changes, merges, overtaking, and human--robot interaction
\cite{LeCleacH2020ALGAMES,WangEtAl2019GameTheoreticPlanning,JiEtAl2021StackelbergLaneMerging,TurnwaldEtAl2016HumanAvoidance}.

A first class of equilibrium-computation methods assumes full-structure model access, where a solver can explicitly formulate the full game model and its associated equilibrium conditions
\cite{Rosen1965,FacchineiPang2003,HarkerPang1990,FerrisMunson2000}.
Under this information structure, NE problems are commonly reformulated as best-response fixed-point problems, variational inequalities (VI), or stacked KKT/complementarity systems, and solved using methods ranging from projected-gradient to constrained dynamic-game solvers such as ALGAMES and DG-SQP
\cite{FacchineiPang2003,HarkerPang1990,LeCleacH2020ALGAMES,ZhuBorrelli2024DGSQP}.
While powerful when the full game model is available, such methods do not directly address settings in which a decision-maker knows its own optimization model but does not have access to the objectives or constraints of other agents.

Distributed NE-seeking methods relax the requirement of centralized access to the full game model by allowing each agent to use its own local objective and feasible set while exchanging decisions, estimates, aggregate quantities, or dual variables over a communication network
\cite{koshal2016distributed,yi2019operator,salehisadaghiani2019distributed,bianchi2022fast}.
Thus, knowledge of other players' objective functions is generally not required. However, these methods typically rely on the agents actively participating in an online equilibrium-seeking protocol through prescribed information exchange.

The information structure considered here is different. Player~2 does not participate in such a distributed equilibrium-seeking protocol and does not reveal its local objective, constraints, gradient, or algorithmic variables. Instead, Player~1 accesses Player~2 through an exact response map or a simulated or learned approximation thereof. Accordingly, the distinction from distributed NE seeking lies in the information and participation structure, rather than in the equilibrium concept being computed.

Other research directions address limited access to an opponent's structural model through inverse optimization, inverse game methods, or learning-based approaches that infer latent objectives or opponent models from interaction data~\cite{AllenEtAl2022IOGNEP,LiuEtAl2023UnknownObjectives,soltanian2025pace}.
The present framework may rely on the same type of response data, but its inferential target is different: it does not require reconstructing Player~2's objective or KKT system and instead uses the response behavior directly as the information interface for equilibrium seeking.

This paper studies a two-player constrained game with decoupled feasible sets under an asymmetric information structure. Player~1 has access to its own objective and constraints, while Player~2's structural model is unavailable to Player~1; instead, Player~1 represents Player~2's response through a best-response correspondence \(\mathcal{B}_2\), a single-valued response map \(\mathrm{BR}_2\), or an approximation \(\widehat{\mathrm{BR}}_2\). Importantly, this does not introduce a new equilibrium concept; throughout the paper, the solution sought is the standard Nash equilibrium of the underlying game. The asymmetry concerns only how Player~2's Nash optimality condition is represented and accessed by Player~1.

Best-response maps are classical in game theory and appear in iterative best-response (IBR), Jacobi, Gauss--Seidel, and gradient-response schemes
\cite{fudenberg1998theory,LeiEtAl2020BR,WilliamsEtAl2018BRMPC}.
Unlike IBR-type methods, which repeatedly solve or approximate players' best-response problems during the equilibrium-seeking iteration, here \(\mathrm{BR}_2\) or \(\widehat{\mathrm{BR}}_2\) serves as the information interface replacing Player~2's unavailable optimization model, while Player~1 updates its decision through projected gradient steps. Although \(x_2=\mathrm{BR}_2(x_1)\) resembles a bilevel or Stackelberg relation, the equilibrium remains Nash because Player~1 optimizes against Player~2's realized action without differentiating through \(\mathrm{BR}_2\).

The response-map viewpoint was motivated in our prior work on interaction-aware autonomy~\cite{RabbaniEtAl2026BRMaps}.
The main analytical contribution of this paper is to show explicitly how Player~1's objective properties and Player~2's response sensitivity jointly determine uniqueness and convergence. In particular, rather than assuming monotonicity and Lipschitz properties for the resulting equilibrium-seeking operator, we derive them from the game structure and quantify how the sensitivity of \(\mathrm{BR}_2\) affects the uniqueness, convergence, and robustness guarantees.

Specifically, the paper makes the following contributions:
\begin{itemize}
    \item We formulate Nash equilibrium seeking under an asymmetric best-response information structure for a class of two-player constrained games with decoupled feasible sets. Player~1 represents Player~2's response through a best-response correspondence and, for uniqueness and algorithmic analysis, a single-valued response map. We establish sufficient conditions for existence and uniqueness of the resulting Nash equilibrium.
    \item We develop a projected gradient-response iteration and establish its global linear convergence under regularity and dominance conditions.
    \item We analyze the inexact-response case in which only an approximate response map \(\widehat{\mathrm{BR}}_2\) is available, and show that a uniform response-map error induces an explicit \(O(\varepsilon)\) ultimate bound around the Nash equilibrium associated with the exact response map.
\end{itemize}

\section{Problem Formulation and Information Structure}
\label{sec:problem}
Consider a two-player game with decoupled feasible sets
\(\mathrm{X}=\mathrm{X}_1\times\mathrm{X}_2\), where
\(x_i\in\mathrm{X}_i\subset\mathbb{R}^{n_i}\).
Player~1 has access to its own cost
\(J_1:\mathrm{X}_1\times\mathrm{X}_2\to\mathbb{R}\) and solves
\[
    \min_{x_1\in\mathrm{X}_1}J_1(x_1,x_2).
\]
Player~2's structural model is unavailable to Player~1 and is represented through a best-response correspondence
\begin{equation*}
    \mathcal{B}_2:\mathrm{X}_1 \rightrightarrows \mathrm{X}_2,
\qquad
\mathcal{B}_2(x_1) \triangleq \arg\min_{x_2\in\mathrm{X}_2} J_2(x_1,x_2),
\end{equation*}
for some objective \(J_2\) whose explicit form need not be known to Player~1.

Let \(N_{\mathrm{X}_1}(x_1)\) denote the normal cone of \(\mathrm{X}_1\) at \(x_1\).
Following the response-map viewpoint of~\cite{RabbaniEtAl2026BRMaps}, since Player~2's optimization model is unavailable to Player~1, Player~2's Nash condition is expressed directly through its best-response correspondence. Thus, the standard Nash equilibrium condition under this information structure is
\begin{equation}
\label{eq:response_map_nash_system}
\begin{aligned}
    0&\in \nabla_{x_1}J_1(x_1^\star,x_2^\star)
    +N_{\mathrm{X}_1}(x_1^\star),\\
    x_2^\star&\in \mathcal{B}_2(x_1^\star).
\end{aligned}
\end{equation}
When Player~2's best response is unique, we denote the resulting single-valued map by \(\mathrm{BR}_2\), and the Nash equilibrium condition \eqref{eq:response_map_nash_system} specializes to
\begin{equation*}
        0\in \nabla_{x_1}J_1(x_1^\star,\mathrm{BR}_2(x_1^\star))
    +N_{\mathrm{X}_1}(x_1^\star).
\end{equation*}

\section{Existence \& Uniqueness of the Equilibrium}
\label{sec:exist_unique}

This section establishes existence and uniqueness of solutions satisfying the Nash equilibrium condition~\eqref{eq:response_map_nash_system}. We first establish existence under general regularity conditions that allow Player~2's best-response to be set-valued, and then impose stronger assumptions under which the best-response is single-valued and Lipschitz continuous.

\begin{assump}{A}
\label{assum:X_sets}
The feasible sets $\mathrm{X}_1$ and $\mathrm{X}_2$ are nonempty, convex, and compact.
\end{assump}

\begin{assump}{A}
\label{assum:J1_cont_convex}
The function $J_1:\mathrm{X}_1\times\mathrm{X}_2\to\mathbb{R}$ is continuous on $\mathrm{X}_1\times\mathrm{X}_2$. Moreover, for every fixed $x_2\in\mathrm{X}_2$, $J_1(\cdot,x_2)$ is differentiable and convex on $\mathrm{X}_1$.
\end{assump}

\begin{assump}{A}
\label{assum:BR2_regular}
For every $x_1\in \mathrm{X}_1$, the correspondence $\mathcal{B}_2(x_1)\subseteq \mathrm{X}_2$ is nonempty, convex, and compact, and $\mathcal{B}_2:\mathrm{X}_1\rightrightarrows \mathrm{X}_2$ is upper hemicontinuous.
\end{assump}

\vspace{0.3em}
\begin{theorem}[Existence under asymmetric response map]
\label{thm:existence}
Under Assumptions~\ref{assum:X_sets}--\ref{assum:BR2_regular}, the Nash equilibrium condition~\eqref{eq:response_map_nash_system} admits at least one solution.
\end{theorem}

\begin{proof}
Define Player~1's best-response correspondence
\begin{equation}
\label{eq:BR1_def_exist}
\mathcal{B}_1(x_2)\triangleq \arg\min_{x_1\in \mathrm{X}_1} J_1(x_1,x_2), \qquad x_2\in \mathrm{X}_2.
\end{equation}
By Assumptions~\ref{assum:X_sets}--\ref{assum:J1_cont_convex}, Weierstrass' theorem and convexity imply that \(\mathcal{B}_1(x_2)\) is nonempty, compact, and convex for every \(x_2\).
Moreover, by Berge's maximum theorem, \(\mathcal{B}_1\) is upper hemicontinuous \cite{Border1985FixedPoint}.
Define the product correspondence $F:\mathrm{X}\rightrightarrows \mathrm{X}$ on $\mathrm{X}=\mathrm{X}_1\times \mathrm{X}_2$ by
\begin{equation*}
\label{eq:F_def}
F(x_1,x_2)\triangleq \mathcal{B}_1(x_2)\times \mathcal{B}_2(x_1).
\end{equation*}
By the preceding arguments and Assumption~\ref{assum:BR2_regular}, for every $(x_1,x_2)\in \mathrm{X}$ the set $F(x_1,x_2)$ is nonempty, convex, and compact, and $F$ is upper hemicontinuous. Since $\mathrm{X}$ is nonempty, convex, and compact by Assumption~\ref{assum:X_sets}, Kakutani's fixed-point theorem~\cite{Border1985FixedPoint} guarantees the existence of $(x_1^\star,x_2^\star)\in \mathrm{X}$ such that
\begin{equation*}
\label{eq:fixed_point}
(x_1^\star,x_2^\star)\in F(x_1^\star,x_2^\star),
\end{equation*}
where $x_1^\star\in \mathcal{B}_1(x_2^\star)$ and $x_2^\star\in \mathcal{B}_2(x_1^\star)$. Therefore $(x_1^\star,x_2^\star)$ satisfies \eqref{eq:response_map_nash_system}.
\end{proof}

\vspace{0.3em}
To state a uniqueness condition aligned with our subsequent convergence analysis, we impose the following stronger assumptions.

\setcounter{innerassump}{0}
\begin{assump}{B}
\label{assum:mu-strongly convex}
Assumption~\ref{assum:J1_cont_convex} holds with $J_1(\cdot,x_2)$ $\mu$-strongly convex on $\mathrm{X}_1$ for every $x_2\in\mathrm{X}_2$, for some $\mu>0$.
\end{assump}

\begin{assump}{B}
\label{assum:cross-dependence lipschitz}
For every fixed $x_1\in \mathrm{X}_1$, the mapping $x_2\mapsto \nabla_{x_1} J_1(x_1,x_2)$ is Lipschitz continuous on $\mathrm{X}_2$ with constant $L_{12}\ge 0$, i.e., $\forall x_2,y_2\in \mathrm{X}_2$,
\begin{equation*}
    \|\nabla_{x_1}J_1(x_1,x_2)-\nabla_{x_1}J_1(x_1,y_2)\|\le L_{12}\|x_2-y_2\|.
\end{equation*}
\end{assump}

\vspace{0.3em}
\begin{assump}{B}
\label{assum:BR2_singleton}
For all $x_1\in \mathrm{X}_1$, the best-response correspondence $\mathcal{B}_2(x_1)$ is single-valued, denoted by \(\mathrm{BR}_2(x_1)\).
\end{assump}

\vspace{0.3em}
\begin{assump}{B}
\label{assum:BR2_lipschitz}
The best-response map \(\mathrm{BR}_2:\mathrm{X}_1\to \mathrm{X}_2\) is Lipschitz continuous on $\mathrm{X}_1$ with constant $L_2\ge 0$.
\end{assump}

Assumption~\ref{assum:BR2_lipschitz} requires Player~2's optimal response to vary continuously at a bounded rate with Player~1's decision. Such Lipschitz regularity is standard for parametric optimization solution maps and holds, for example, under uniform strong convexity and suitable parameter dependence, as well as for linear-quadratic problems with a unique optimizer~\cite{BonnansShapiro2013}.

\vspace{0.3em}
\begin{theorem}[Uniqueness under asymmetric response map]
\label{thm:uniqueness}
Suppose Assumptions~\ref{assum:X_sets} and \ref{assum:mu-strongly convex}--\ref{assum:BR2_lipschitz} hold. If $\mu>L_{12}L_2$, then the Nash equilibrium condition associated with \(\mathrm{BR}_2\) has a unique solution.
\end{theorem}

\vspace{0.3em}
\begin{proof}
Assumption~\ref{assum:mu-strongly convex} implies Assumption~\ref{assum:J1_cont_convex}, while Assumptions~\ref{assum:BR2_singleton}--\ref{assum:BR2_lipschitz} imply that \(\mathcal{B}_2(x_1)=\{\mathrm{BR}_2(x_1)\}\) is nonempty, compact, convex, and upper hemicontinuous, since \(\mathrm{BR}_2\) is Lipschitz continuous. Hence Assumption~\ref{assum:BR2_regular} holds in the single-valued case, and existence follows from Theorem~\ref{thm:existence}.

To prove uniqueness, let $(a_1,b_1)$ and $(a_2,b_2)$ be two Nash equilibria. By Assumption~\ref{assum:BR2_singleton},
\begin{equation}
\label{eq:NE_pairs}
b_1=\mathrm{BR}_2(a_1), \qquad b_2=\mathrm{BR}_2(a_2).
\end{equation}
Since $a_1$ minimizes $J_1(\cdot,b_1)$ over the convex set $\mathrm{X}_1$, the first-order optimality condition for convex constrained problems~\cite{FacchineiPang2003} yields
\[
\langle \nabla_{x_1}J_1(a_1,b_1),\, a_2-a_1\rangle \ge 0.
\]
Similarly, optimality of $a_2$ for $J_1(\cdot,b_2)$ implies
\[
\langle \nabla_{x_1}J_1(a_2,b_2),\, a_1-a_2\rangle \ge 0.
\]
Adding the two inequalities gives
\begin{equation}
\label{eq:FOC}
\langle \nabla_{x_1}J_1(a_1,b_1)-\nabla_{x_1}J_1(a_2,b_2),\, a_1-a_2\rangle \le 0.
\end{equation}

Add and subtract $\nabla_{x_1}J_1(a_2,b_1)$ inside the inner product in \eqref{eq:FOC}. By $\mu$-strong convexity of $J_1(\cdot,b_1)$ (Assumption~\ref{assum:mu-strongly convex}), we have
\begin{equation}
\label{eq:mu_term}
\begin{aligned}
\langle \nabla_{x_1}J_1(a_1,b_1)-\nabla_{x_1}J_1(a_2,b_1),\, &a_1-a_2\rangle
\\&\ge \mu\|a_1-a_2\|^2.
\end{aligned}
\end{equation}
For the remaining term, Cauchy--Schwarz yields
\begin{equation*}
    \begin{aligned}
        \langle \nabla_{x_1}&J_1(a_2,b_1)-\nabla_{x_1}J_1(a_2,b_2),\, a_1-a_2\rangle
\\ &\ge -\|\nabla_{x_1}J_1(a_2,b_1)-\nabla_{x_1}J_1(a_2,b_2)\|\,\|a_1-a_2\|.
    \end{aligned}
\end{equation*}
Using Assumption~\ref{assum:cross-dependence lipschitz} and \eqref{eq:NE_pairs} together with Assumption~\ref{assum:BR2_lipschitz}, we obtain
\begin{equation}
\label{eq:cross_term}
\begin{aligned}
    \langle \nabla_{x_1}J_1(a_2,b_1)-\nabla_{x_1}&J_1(a_2,b_2),\, a_1-a_2\rangle
\\& \ge -L_{12}\|b_1-b_2\|\,\|a_1-a_2\|
\\& \ge -L_{12}L_2\|a_1-a_2\|^2.
\end{aligned}
\end{equation}
Combining \eqref{eq:FOC}--\eqref{eq:cross_term} yields
\[
(\mu-L_{12}L_2)\|a_1-a_2\|^2 \le 0.
\]
If $\mu>L_{12}L_2$, then $a_1=a_2$. Substituting into \eqref{eq:NE_pairs}, and given Assumption~\ref{assum:BR2_singleton}, gives $b_1=b_2$, hence $(a_1,b_1)=(a_2,b_2)$. Therefore, the Nash equilibrium associated with the exact response map \(\mathrm{BR}_2\) is unique.
\end{proof}

The same dominance condition $\mu>L_{12}L_2$ also induces a monotonicity margin that will be central to the convergence analysis of the iterative scheme developed next.

\section{Projected Gradient--Response Iteration \& Global Linear Convergence}
\label{sec:alg_conv}

We now turn to computation and analyze a projected gradient-response iteration for the single-valued response-map setting.

\vspace{0.3em}
\setcounter{innerassump}{0}
\begin{assump}{C}
\label{assum:J1 Lipschitz on x1}
For every fixed $x_2\in \mathrm{X}_2$, the function $J_1(\cdot,x_2)$ is $L_1$-smooth on $\mathrm{X}_1$, i.e., $\forall x_1,y_1\in \mathrm{X}_1$,
\begin{equation*}
\|\nabla_{x_1}J_1(x_1,x_2)-\nabla_{x_1}J_1(y_1,x_2)\|\le L_1\|x_1-y_1\|.
\end{equation*}
\end{assump}

\vspace{0.5em}
The reduced normal-cone condition is equivalently a projected fixed point~\cite{FacchineiPang2003}, which motivates the following projected gradient-response iteration:
\begin{equation}
\label{eq:iterations}
\begin{aligned}
x_2^k &= \mathrm{BR}_2(x_1^k),\\
x_1^{k+1} &= \Proj{x_1^k - \alpha \nabla_{x_1}J_1(x_1^k,x_2^k)},
\end{aligned}
\end{equation}
where $\Pi_{\mathrm{X}_1}(\cdot)$ denotes the Euclidean projection onto $\mathrm{X}_1$.
Since $\mathrm{X}_1$ is closed and convex (Assumption~\ref{assum:X_sets}), the projection is nonexpansive:
\begin{equation}
\label{eq:nonexpansive projection}
    \|\Proj{u}-\Proj{v}\| \le \|u-v\|, \qquad \forall u,v\in\mathbb{R}^{n_1}.
\end{equation}

\vspace{0.3em}
\begin{theorem}[Global linear convergence]
\label{thm:linear_conv}
Suppose Assumptions~\ref{assum:X_sets}, \ref{assum:mu-strongly convex}--\ref{assum:BR2_lipschitz}, and \ref{assum:J1 Lipschitz on x1} hold.
Let
\begin{equation}
\label{eq:definition of m and LG}
    m \triangleq \mu - L_{12}L_2 \quad\text{and}\quad L_G \triangleq L_1 + L_{12}L_2,
\end{equation}
and note that $m>0$ follows from $\mu>L_{12}L_2$.
Then, for any stepsize $\alpha$ satisfying
\begin{equation*}
\label{eq:alpha bound}
0<\alpha< \frac{2m}{L_G^2}
=\frac{2(\mu-L_{12}L_2)}{(L_1+L_{12}L_2)^2},
\end{equation*}
the iterates generated by \eqref{eq:iterations} satisfy
\[
\|x_1^k-x_1^\star\|\le \rho(\alpha)^k \|x_1^0-x_1^\star\|,
\]
\[
\|x_2^k-x_2^\star\|\le L_2\,\rho(\alpha)^k\|x_1^0-x_1^\star\|,
\]
where $(x_1^\star,x_2^\star)$ is the unique Nash equilibrium associated with the response map \(\mathrm{BR}_2\), and
\[
\rho(\alpha)\triangleq \sqrt{1-2\alpha m+\alpha^2L_G^2}\in(0,1).
\]
\end{theorem}

\vspace{0.5em}
\begin{proof}
Define the map $T:\mathrm{X}_1\to\mathrm{X}_1$ by
\begin{equation}
    \label{eq:Tmap}
T(x_1) := \Proj{x_1 - \alpha \nabla_{x_1}J_1\big(x_1,\mathrm{BR}_2(x_1)\big)},
\end{equation}
and define the operator $G:\mathrm{X}_1\to\mathbb{R}^{n_1}$ as
\[
G(x_1) := \nabla_{x_1}J_1\big(x_1,\mathrm{BR}_2(x_1)\big).
\]
For any $x_1,y_1\in\mathrm{X}_1$, add and subtract $\nabla_{x_1}J_1(y_1,\mathrm{BR}_2(x_1))$ to obtain
\begin{equation*}
    \begin{aligned}
        \|G(&x_1)-G(y_1)\|\\
        &\le
        \|\nabla_{x_1}J_1(x_1,\mathrm{BR}_2(x_1))-\nabla_{x_1}J_1(y_1,\mathrm{BR}_2(x_1))\| \\
        &\quad+
        \|\nabla_{x_1}J_1(y_1,\mathrm{BR}_2(x_1))-\nabla_{x_1}J_1(y_1,\mathrm{BR}_2(y_1))\|.
    \end{aligned}
\end{equation*}
By Assumption~\ref{assum:J1 Lipschitz on x1}, the first term is bounded by $L_1\|x_1-y_1\|$.
By Assumptions~\ref{assum:cross-dependence lipschitz} and~\ref{assum:BR2_lipschitz}, the second term is bounded by
$L_{12}\|\mathrm{BR}_2(x_1)-\mathrm{BR}_2(y_1)\|\le L_{12}L_2\|x_1-y_1\|$.
Therefore, $G$ is Lipschitz continuous on $\mathrm{X}_1$ with constant $L_G=L_1+L_{12}L_2$, i.e.,
\begin{equation}
    \label{eq:LG}
\|G(x_1)-G(y_1)\|\le L_G\|x_1-y_1\|.
\end{equation}

Next, using the same add--subtract decomposition and Assumption~\ref{assum:mu-strongly convex}, we have
\begin{equation*}
    \begin{aligned}
        \langle &G(x_1)-G(y_1),\,x_1-y_1\rangle\\
        &=
        \langle \nabla_{x_1}J_1(x_1,\mathrm{BR}_2(x_1))-\nabla_{x_1}J_1(y_1,\mathrm{BR}_2(x_1)),\,x_1-y_1\rangle \\
        &\,+
        \langle \nabla_{x_1}J_1(y_1,\mathrm{BR}_2(x_1))-\nabla_{x_1}J_1(y_1,\mathrm{BR}_2(y_1)),\,x_1-y_1\rangle.
    \end{aligned}
\end{equation*}
The first inner product is lower bounded by $\mu\|x_1-y_1\|^2$. For the second, Cauchy--Schwarz together with Assumptions~\ref{assum:cross-dependence lipschitz} and~\ref{assum:BR2_lipschitz} yields
\begin{equation*}
    \begin{aligned}
        \langle \nabla_{x_1}J_1(y_1,\mathrm{BR}_2(x_1))-\nabla_{x_1}J_1(y_1,& \mathrm{BR}_2(y_1)),\,x_1-y_1\rangle \\
&\ge -L_{12}L_2\|x_1-y_1\|^2.
    \end{aligned}
\end{equation*}
Thus, $G$ is strongly monotone on $\mathrm{X}_1$ with constant $m=\mu-L_{12}L_2>0$, i.e.,
\begin{equation}
    \label{eq:m-strong G}
\langle G(x_1)-G(y_1),\,x_1-y_1\rangle \ge m\|x_1-y_1\|^2.
\end{equation}

By nonexpansiveness of the projection \eqref{eq:nonexpansive projection},
\[
\|T(x_1)-T(y_1)\|
\le \|(x_1-y_1)-\alpha(G(x_1)-G(y_1))\|.
\]
Expanding the square and applying the Lipschitz and strong monotonicity bounds in \eqref{eq:LG} and \eqref{eq:m-strong G} gives
\begin{equation}
\label{eq:upperbound for T}
    \begin{aligned}
            \|T(x_1)-T(y_1)\|^2
\le \big(1-2&\alpha m+\alpha^2L_G^2\big)\|x_1-y_1\|^2\\
& \quad =: q(\alpha)\|x_1-y_1\|^2.
    \end{aligned}
\end{equation}
For $0<\alpha<2m/L_G^2$, we have $q(\alpha)\in(0,1)$, hence $T$ is a contraction with modulus $\rho(\alpha)=\sqrt{q(\alpha)}\in(0,1)$.
Therefore, by the Banach fixed-point theorem, $T$ has a unique fixed point $x_1^*\in\mathrm{X}_1$ and the iterates satisfy
\begin{equation}
    \label{eq:global_linear_x1}
\|x_1^k-x_1^*\|\le \rho(\alpha)^k\|x_1^0-x_1^*\|.
\end{equation}

Let $x_2^*:=\mathrm{BR}_2(x_1^*)$. Since $\mathrm{BR}_2$ is Lipschitz (\ref{assum:BR2_lipschitz}),
\[
\norm{x_2^k-x_2^*}=\norm{\mathrm{BR}_2(x_1^k)-\mathrm{BR}_2(x_1^*)}\leq L_2\,\norm{x_1^k-x_1^*}.
\]

Combining with \eqref{eq:global_linear_x1} gives
\[
\norm{x_2^k-x_2^*}\leq L_2\,\rho(\alpha)^k\norm{x_1^0-x_1^*}.
\]

Finally, \(x_1^\ast=T(x_1^\ast)\) implies, by the projection optimality condition~\cite{FacchineiPang2003},
\[
\langle \nabla_{x_1}J_1(x_1^\ast,x_2^\ast),x-x_1^\ast\rangle\ge0,\quad \forall x\in X_1.
\]
Since \(J_1(\cdot,x_2^\ast)\) is convex and differentiable, \(x_1^\ast\in\arg\min_{x_1\in X_1}J_1(x_1,x_2^\ast)\). Together with \(x_2^\ast=\mathrm{BR}_2(x_1^\ast)\), this satisfies \eqref{eq:response_map_nash_system}; uniqueness follows from Theorem~\ref{thm:uniqueness}.
\end{proof}

Intuitively, the sufficient condition \(\mu>L_{12}L_2\) requires Player~1's objective to have enough curvature in its own decision so that the descent direction is not dominated by changes in Player~2's response. Thus, although Player~2 reacts to Player~1, this reaction must not perturb Player~1's gradient strongly enough to destroy contraction.

\section{Robustness to Inexact Response Maps}
\label{sec:inexact_br}
The analysis in Section~\ref{sec:alg_conv} assumes that Player~1 can evaluate the exact response map $\mathrm{BR}_2$, without requiring access to Player~2's internal optimization model. In practice, the response map may instead be estimated from simulation queries, observed interactions, or system-identification methods. Such approximations introduce response-model error, which can perturb the equilibrium-seeking iteration. We therefore quantify how bounded response-map error affects convergence to the exact Nash equilibrium.

Let $\widehat{\mathrm{BR}}_2:\mathrm{X}_1\to\widetilde{\mathrm{X}}_2$ denote an approximation of $\mathrm{BR}_2$, where $\widetilde{\mathrm{X}}_2\supseteq\mathrm{X}_2$. Assume that $\nabla_{x_1}J_1$ is defined on $\mathrm{X}_1\times\widetilde{\mathrm{X}}_2$ and that the Lipschitz condition in Assumption~\ref{assum:cross-dependence lipschitz} holds on $\widetilde{\mathrm{X}}_2$ with the same constant $L_{12}$. Further, assume that for some $\varepsilon\geq 0$,
\begin{equation}
\label{eq:approximate best response}
    \|\widehat{\mathrm{BR}}_2(x_1)-\mathrm{BR}_2(x_1)\|\le \varepsilon, \quad \forall x_1\in \mathrm{X}_1.
\end{equation}

We consider the inexact variant of \eqref{eq:iterations}:
\begin{equation}
\label{eq:iterations_inexact}
\begin{aligned}
        &x_2^k = \widehat{\mathrm{BR}}_2(x_1^k), \\
        &x_1^{k+1} = \Proj{x_1^k - \alpha \nabla_{x_1}J_1(x_1^k,x_2^k)}.
\end{aligned}
\end{equation}

\vspace{0.3em}
\begin{theorem}[Robustness to inexact response map]
\label{thm:robust_inexact}
Let $(x_1^\star,x_2^\star)\in \mathrm{X}$ denote the unique Nash equilibrium associated with the exact map $\mathrm{BR}_2$ under Assumptions~\ref{assum:X_sets}, \ref{assum:mu-strongly convex}--\ref{assum:BR2_lipschitz}, and \ref{assum:J1 Lipschitz on x1}, and suppose $\mu>L_{12}L_2$. Take parameters as in \eqref{eq:definition of m and LG} and choose $\alpha\in(0,2m/L_G^2)$. If $\widehat{\mathrm{BR}}_2$ satisfies \eqref{eq:approximate best response}, then the iterates \eqref{eq:iterations_inexact} satisfy, for all $k\ge 0$,
\begin{equation}
\label{eq:robust_recursion}
\|x_1^{k+1}-x_1^\star\|\le \rho(\alpha)\|x_1^k-x_1^\star\| + \alpha L_{12}\varepsilon,
\end{equation}
where
\[
\rho(\alpha)=\sqrt{1-2\alpha m+\alpha^2L_G^2}\in(0,1).
\]
Consequently,
\begin{equation}
\label{eq:limsup_x1}
\limsup_{k\to\infty}\|x_1^k-x_1^\star\|
\le
\frac{\alpha L_{12}}{1-\rho(\alpha)}\,\varepsilon.
\end{equation}
Moreover, the sequence $x_2^k=\widehat{\mathrm{BR}}_2(x_1^k)$ satisfies
\begin{equation}
\label{eq:x2_bound}
\|x_2^k-x_2^\star\|
\le
L_2\|x_1^k-x_1^\star\|+\varepsilon,
\end{equation}
and hence
\begin{equation}
\label{eq:limsup_x2}
\limsup_{k\to\infty}\|x_2^k-x_2^\star\|
\le
\big(L_2\frac{\alpha L_{12}}{1-\rho(\alpha)}+1\big)\varepsilon.
\end{equation}
\end{theorem}

\vspace{0.5em}
\begin{proof}
Let $T:\mathrm{X}_1\to\mathrm{X}_1$ denote the exact map from \eqref{eq:Tmap} and define the inexact map $\widehat{T}:\mathrm{X}_1\to\mathrm{X}_1$ by
\begin{equation*}
\label{eq:That_def}
\widehat{T}(x_1) := \Proj{x_1 - \alpha \nabla_{x_1}J_1\big(x_1,\widehat{\mathrm{BR}}_2(x_1)\big)}.
\end{equation*}
By Theorem~\ref{thm:linear_conv}, $T$ is a contraction on $\mathrm{X}_1$ with modulus $\rho(\alpha)\in(0,1)$, i.e. \eqref{eq:upperbound for T}.
Moreover, using nonexpansiveness of the projection \eqref{eq:nonexpansive projection}, Assumption~\ref{assum:cross-dependence lipschitz}, and \eqref{eq:approximate best response}, we obtain, for any $x_1\in\mathrm{X}_1$,
\begin{equation}
\label{eq:T_hat_minus_T}
    \begin{aligned}
    \|\widehat{T}&(x_1)-T(x_1)\|\\
    &\le \alpha L_{12}\|\widehat{\mathrm{BR}}_2(x_1)-\mathrm{BR}_2(x_1)\|
    \le \alpha L_{12}\varepsilon.
    \end{aligned}
\end{equation}

Now, since $x_1^{k+1}=\widehat{T}(x_1^k)$ and $x_1^\star=T(x_1^\star)$, we have
\begin{equation*}
    \begin{aligned}
\|x_1^{k+1}-x_1^\star\|
&\le \|\widehat{T}(x_1^k)-T(x_1^k)\|+\|T(x_1^k)-T(x_1^\star)\|.
    \end{aligned}
\end{equation*}

Using the contraction property of $T$ from Theorem~\ref{thm:linear_conv} together with \eqref{eq:T_hat_minus_T}, immediately proves \eqref{eq:robust_recursion}.
Iterating \eqref{eq:robust_recursion} yields
\[
\|x_1^k-x_1^\star\|
\le
\rho(\alpha)^k\|x_1^0-x_1^\star\|
+
\alpha L_{12}\varepsilon\sum_{i=0}^{k-1}\rho(\alpha)^i.
\]
Since $\rho(\alpha)\in(0,1)$,
\[
\sum_{i=0}^{k-1}\rho(\alpha)^i
=
\frac{1-\rho(\alpha)^k}{1-\rho(\alpha)},
\]
and therefore, for all $k\ge 0$,
\begin{equation*}
\label{eq:robust_finite_bound}
\|x_1^k-x_1^\star\|
\le
\rho(\alpha)^k\|x_1^0-x_1^\star\|
+
\frac{1-\rho(\alpha)^k}{1-\rho(\alpha)}\,\alpha L_{12}\varepsilon.
\end{equation*}
Taking $\limsup$ as $k\to\infty$ gives \eqref{eq:limsup_x1}.

Finally, since $x_2^k=\widehat{\mathrm{BR}}_2(x_1^k)$ and $x_2^\star=\mathrm{BR}_2(x_1^\star)$, we write
\begin{equation*}
    \begin{aligned}
    \|x_2^k-x_2^\star\|
    &\le \|\widehat{\mathrm{BR}}_2(x_1^k)-\mathrm{BR}_2(x_1^k)\|\\
    & \quad +\|\mathrm{BR}_2(x_1^k)-\mathrm{BR}_2(x_1^\star)\|.
    \end{aligned}
\end{equation*}
The two terms are bounded by \(\varepsilon\) and \(L_2\|x_1^k-x_1^\star\|\), respectively; hence \eqref{eq:x2_bound} follows, and \eqref{eq:limsup_x2} follows from \eqref{eq:limsup_x1}.
\end{proof}

Thus, a uniformly bounded response-map error induces
an explicit \(O(\varepsilon)\) ultimate bound around
the Nash equilibrium associated with the exact
response map.

\section{Numerical Example}
\label{sec:example}

We illustrate the results on a one-dimensional tug-of-war cart over a horizon \(N=15\), where Player~1 chooses an open-loop force sequence and Player~2 reacts through a response map. The example validates exact convergence, inexact-response robustness, and the predicted \(O(\varepsilon)\) scaling.

Let $z_t=[p_t,\;\dot p_t]^\top\in\mathbb{R}^2$ denote position and velocity of the cart at time step $t$. With sampling time $\Delta t$, mass $m$, and viscous damping coefficient $c$, the discrete-time dynamics are
\begin{equation*}
\label{eq:ex_dyn}
z_{t+1}=A_d z_t + B_u u_t + B_v v_t,\qquad t=0,\ldots,N-1,
\end{equation*}
where $z_0=[0,0]^\top$ and
\[
A_d=\begin{bmatrix}1&\Delta t\\0&1-\Delta t(c/m)\end{bmatrix},\qquad
B_u=-B_v=\begin{bmatrix}0\\ \Delta t/m\end{bmatrix}.
\]
Stacking the horizon variables as $u=[u_0,\ldots,u_{N-1}]^\top$, and $v=[v_0,\ldots,v_{N-1}]^\top$,
results in the predicted state trajectory of the affine form
\[
Z(u,v)=S_0+S_u u+S_v v.
\]

Player~1 selects $u\in\mathrm{X}_1$ to minimize
\begin{equation*}
\label{eq:ex_J1}
\begin{aligned}
J_1(u,v)
=
\frac{q_u}{2}\|u\|_2^2
&+\frac{1}{2} Z(u,v)^\top Q Z(u,v)\\
&-r_{\mathrm{pull}}\, C_x^\top Z(u,v)
+\gamma\,u^\top v,
\end{aligned}
\end{equation*}
where $Q=\mathrm{blkdiag}(q_{\mathrm{pos}},q_{\mathrm{vel}},\ldots,q_{\mathrm{pos}},q_{\mathrm{vel}})\succeq0$ penalizes predicted states, and $C_x$ selects position components so that $C_x^\top Z(u,v)=\sum_{t=1}^{N} p_t$. The feasible set of Player~1 is the box
\begin{equation*}
\label{eq:ex_X1}
\mathrm{X}_1 := \{u\in\mathbb{R}^{N}\mid u_{\min}\le u_t\le u_{\max},\; t=0,\ldots,N-1\}.
\end{equation*}

Player~2 has the separable convex objective
\[
J_2(u,v)=\sum_t \left[\frac{1}{\kappa}\int_0^{v_t}\operatorname{atanh}\left(\frac{s}{v_{\max}}\right)ds-(\bar v+c_{\mathrm{react}}u_t)v_t\right]
\]
over $\mathrm{X}_2=\{v\in\mathbb{R}^N\mid\|v\|_\infty\le v_{\max}\}$. Although $J_2$ is unavailable to Player~1, its unique minimizer results in the stagewise response
\[
v_t=b(u_t)=v_{\max}\tanh\!\big(\kappa(\bar v+c_{\mathrm{react}}u_t)\big),
\]
yielding
\[
\mathrm{BR}_2(u)=[b(u_0),\ldots,b(u_{N-1})]^\top\in\mathrm{X}_2.
\]
Thus, Player~1 requires only $\mathrm{BR}_2$ or an approximation thereof, rather than $J_2$ itself.

For theoretical verification, $(\mu,L_1,L_{12},L_2)$ are computed analytically from Player~1's cost and the exact response map to certify the convergence condition and guaranteed stepsize range. However, the update itself requires only Player~1's gradient and response-map evaluations.
\begin{figure}[t]
    \centering
    \includegraphics[width=1\linewidth]{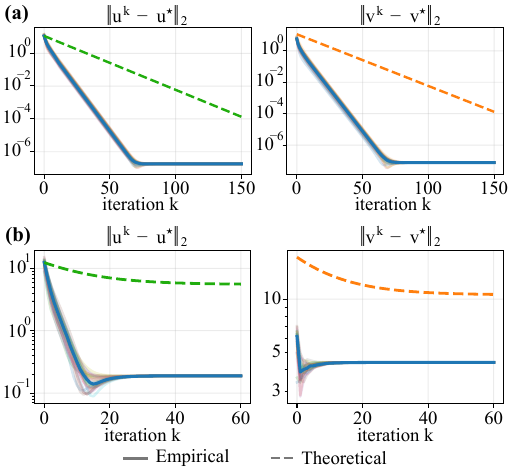}
    \caption{(a) Exact response-map iteration from multiple random initializations, showing linear convergence. (b) Inexact response-map iteration showing convergence to an $\varepsilon$-dependent neighborhood. Empirical trajectories are shown with the corresponding theoretical bounds.}
    \label{fig:exact_conv_uv}
    \vspace{-10pt}
\end{figure}
We use \(N=15\), \(|u_t|\leq3\),
\((\Delta t,m,c)=(0.2,2.0,2.0)\),
\((q_u,q_{\mathrm{pos}},q_{\mathrm{vel}},r_{\mathrm{pull}},\gamma)
=(0.35,0.32,0.01,5.4,0.01)\), and
\((v_{\max},\kappa,\bar v,c_{\mathrm{react}})
=(1.8,0.6,0.9,0.9)\), which yield
\[
(\mu,L_1,L_{12},L_2)=(0.3500,0.4857,0.1257,0.9720).
\]
Thus, \(\mu-L_{12}L_2=0.2279>0\) and \(\alpha_{\max}=1.2336\).

We first consider the exact iteration \eqref{eq:iterations} with $\alpha=0.5\,\alpha_{\max}$, which gives $\rho(\alpha)=0.9271<1$. Fig.~\ref{fig:exact_conv_uv}(a) shows $\|u^k-u^\star\|_2$ and $\|v^k-v^\star\|_2$ from multiple random initializations, together with the theoretical decay references from Theorem~\ref{thm:linear_conv}. A reference equilibrium $(u^\star,v^\star)$ is computed by iterating the exact map to a tight tolerance. The observed trajectories exhibit the predicted linear convergence, while the late-stage flattening is due to numerical tolerance.

To emulate imperfect response-map information, we replace the stage-response function \(b(\cdot)\) by its first-order Taylor expansion about \(u_t=0\):
\[
\widehat{b}(u_t)
=
v_{\max}\tanh(\kappa \bar v)
+
v_{\max}\kappa c_{\mathrm{react}}
\mathrm{sech}^2(\kappa \bar v)\,u_t .
\]
Define \(\widehat{\mathrm{BR}}_2(u)=[\widehat b(u_0),\ldots,\widehat b(u_{N-1})]^\top\). With
\[
\varepsilon_{\mathrm{stage}}
:=
\sup_{u_t\in[u_{\min},u_{\max}]}
|\widehat b(u_t)-b(u_t)|,
\]
the induced uniform error satisfies
\[
\|\widehat{\mathrm{BR}}_2(u)-\mathrm{BR}_2(u)\|_2
\le
\sqrt{N}\,\varepsilon_{\mathrm{stage}}
=:\varepsilon,
\qquad
\forall u\in\mathrm{X}_1 .
\]

We then run the inexact iteration \eqref{eq:iterations_inexact} with the same stepsize. Fig.~\ref{fig:exact_conv_uv}(b) shows that the iterates no longer converge to the exact Nash equilibrium, but instead enter and remain in a neighborhood, in agreement with Theorem~\ref{thm:robust_inexact}. The corresponding theoretical bounds are shown in Fig.~\ref{fig:exact_conv_uv}(b).
The bound on $v^k$ is conservative because \eqref{eq:x2_bound} contains the additive approximation term $\varepsilon$ directly.

Finally, to validate the predicted $O(\varepsilon)$ scaling, we inject a constant perturbation $d\in\mathbb{R}^N$ into the response-map approximation, i.e., $\widehat{\mathrm{BR}}_2(u)=\mathrm{BR}_2(u)+d$ with $\|d\|_2=\varepsilon$, and estimate
\[
\Delta_u(\varepsilon)\triangleq \limsup_{k\to\infty}\|u^k-u^\star\|_2,
\]
using a tail-limsup over the final iterations. Fig.~\ref{fig:fig_neighborhood_vs_epsilon} shows that the measured steady-state deviation scales approximately linearly with $\varepsilon$ and remains below the theoretical envelopes
\[
R_u(\varepsilon)=\frac{\alpha L_{12}}{1-\rho(\alpha)}\,\varepsilon,
\qquad
R_v(\varepsilon)=L_2R_u(\varepsilon)+\varepsilon.
\]

\begin{figure}
    \centering
    \includegraphics[width=1\linewidth]{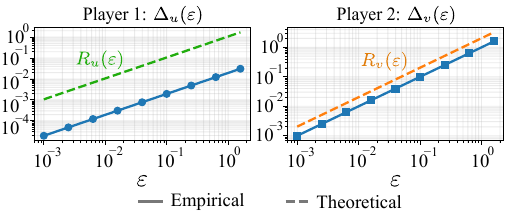}
    \caption{\(O(\varepsilon)\) scaling of the steady-state deviation under bounded response-map error.}
    \label{fig:fig_neighborhood_vs_epsilon}
    \vspace{-10pt}
\end{figure}

\section{Conclusion}
This paper studied Nash seeking under asymmetric response-map information for constrained two-player games, where Player~1 accesses Player~2's behavior through a response map rather than Player~2's explicit optimization model. We established existence under set-valued best-response information, uniqueness under a single-valued Lipschitz response map, global linear convergence of the projected gradient-response iteration, and $O(\varepsilon)$ robustness under uniformly bounded response-map error. Future work will consider multi-player and coupled-constraint settings, convergence under set-valued or stochastic responses, and less conservative convergence conditions.

\bibliographystyle{IEEEtran}
\bibliography{reference,cdc_final_addition}

\end{document}